\documentclass[10pt]{amsproc}
\usepackage{latexsym,amsmath, amsthm, amssymb, amsfonts,tikz}
\usepackage[color]{optional} 
\opt{color}{\definecolor{mygreen}{RGB}{0,153,0}}
\opt{gray}{\definecolor{mygreen}{gray}{0}}
\opt{color}{\definecolor{myred}{RGB}{255,0,0}}
\opt{gray}{\definecolor{myred}{gray}{0}}
\opt{color}{\definecolor{myblue}{RGB}{0,0,255}}
\opt{gray}{\definecolor{myblue}{gray}{0}}
\opt{color}{\definecolor{myorange}{RGB}{255,128,0}}
\opt{gray}{\definecolor{myorange}{gray}{0}}

\renewcommand{\cases}[1]    {\left\{ \begin{array}{ll}#1\end{array}\right.}
\newcommand{\norm}[1]       {\left\| #1\right\|}

\newcounter{theorem}
\newtheorem*{theorem}        {Theorem}
\newtheorem*{definition} {Definition}
\newtheorem{lemma} [theorem] {Lemma}
\newtheorem*{problem}{Hole Problem}

\newcommand{\figlabel}[1]   {\label{fig:#1}}
\newcommand{\lemlabel}[1]   {\label{lem:#1}}
\newcommand{\thmlabel}[1]   {\label{thm:#1}}

\newcommand{\figref}[1]     {Figure~\ref{fig:#1}}
\newcommand{\lemref}[1]     {Lemma~\ref{lem:#1}}
\newcommand{\secref}[1]     {Section~\ref{sec:#1}}

\title[Filling a Hole in a Crease Pattern]{Filling a Hole in a Crease Pattern: \\ Isometric Mapping from Prescribed Boundary Folding}
\author[Erik D. Demaine\qquad Jason S. Ku]{Erik D. Demaine$^*$\qquad Jason S. Ku$^\dagger$}
\thanks{$^*$MIT Computer Science and Artificial Intelligence Laboratory, 32 Vassar Street, Cambridge, MA 02139 USA, \texttt{edemaine@mit.edu}}
\thanks{$^\dagger$MIT Field Intelligence Laboratory, 77 Massachusetts Avenue, Cambridge, MA 02139, USA, \texttt{jasonku@mit.edu}}
\date{}

\begin{document}
\begin{abstract}
Given a sheet of paper and a prescribed folding of its boundary, is there a way to fold the paper's interior without stretching so that the boundary lines up with the prescribed boundary folding? For polygonal boundaries nonexpansively folded at finitely many points, we prove that a consistent isometric mapping of the polygon interior always exists and is computable in polynomial time.
\end{abstract}
\maketitle

\section{Introduction}
Many problems in origami require the folder to map the perimeter of a piece of paper to some specified folded configuration. In the tree method of origami design, circle packing breaks the paper up into polygonal molecules whose perimeter must be mapped to a specific tree. The fold-and-cut problem inputs a set of polygonal silhouettes whose perimeters must be mapped onto a common line. These two problems are well studied; one solution to the molecule folding problem is the universal molecule \cite{Lang} while a solution to the fold-and-cut problem lies in the polygon's straight skeleton \cite{JCDCG98}. Both of these problems can be considered as specific versions of a more general problem: the hole problem. 

Given a crease pattern with a hole in it (an area of the paper with the creases missing), can we fill in the hole with suitable creases? More precisely, given a sheet of paper and a prescribed folding of its boundary, is there a way to fold the paper's interior without stretching so that the boundary lines up with the prescribed boundary folding? This hole problem was originally proposed by Barry Hayes at 3OSME in 2001 with the motivation of finding flat-foldable gadgets with common interfaces satisfying certain properties, such as not-all-equal clauses for an NP-hardness reduction \cite{Bern}.

\begin{figure}[t]
\opt{color}{\includegraphics[width=4.5in]{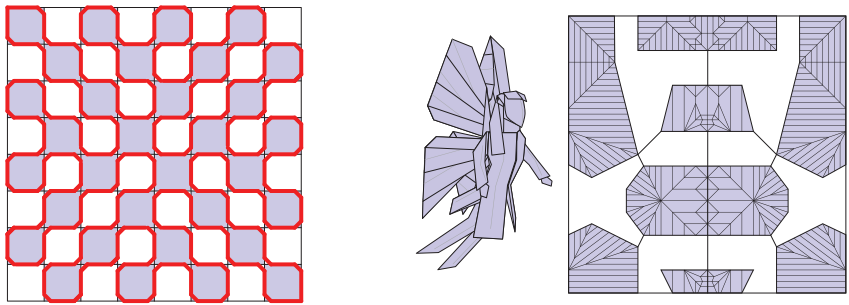}}
\opt{gray}{\includegraphics[width=4.5in]{img/design_bw.eps}}
\caption{(Left) A boundary mapping that might be used to design a color-change checker board model. (Right) An unfinished crease pattern with parts of the crease pattern unknown.}
\figlabel{fig:design}
\end{figure}

This problem formulation can be transformed to solve several existing problems, as well as some new applications (see \figref{fig:design}). If we map the boundary to a line, the polygon is now a molecule to be filled with creases or one half of a fold-and-cut problem cutline. The hole problem can also address problems where the boundary is not mapped to a line, i.e. mappings into the plane or into three dimensions, potentially leading to the algorithmic design of multi-axial bases, color-changes, or complex three-dimensional tessellation or modulars. When trying to combine separately designed parts of an origami model, a solution to the hole problem could be used to design an interfacing crease pattern between them.

In this paper, we show that the hole problem always has a solution for polygonal input boundaries folded at finitely many points under the obvious necessary condition that the input folding is nonexpansive, and present a polynomial-time algorithm to find one. We restrict ourselves to isometry and ignore self-intersection, leaving layer ordering (if possible) as an open problem. \secref{notation} introduces notation and defines the problem. \secref{necessary} discusses the necessary condition which will turn out to be sufficient. \secref{bend} constructs vertex creases satisfying local isometry. \secref{split} propagates the creases. \secref{partition} describes partitioning polygons. \secref{theorem} describes the algorithm. \secref{application} discusses application and implementation. \secref{conclusion} summarizes the results.

\section{Notation and Definitions}
\label{sec:notation}

First some notation and definitions. Let $\norm{\cdot}$ denote Euclidean distance. Given a set of points $A\subseteq B\subset \mathbb{R}^c, c\in\mathbb{Z}^+$ and mapping $f:B\rightarrow\mathbb{R}^d, d\in\mathbb{Z}^+$, we say that $A$ is (\emph{expansive}, \emph{contractive}, \emph{critical}) under $f$ if $\norm{u-v} (<,>,=) \norm{f(u)-f(v)}$ for every $u,v\in A$, with (\emph{nonexpansive}, \emph{noncontractive}, \emph{noncritical}) referring to respective negations. Critical is the same as isometric under the Euclidean metric, but because we will use the term ``isometry" to refer to isometric maps under the shortest-path metric \cite{GFALOP}, we use a different term for clarity. We say two line segments \emph{cross} if their intersection is nonempty. We now prove two relations on crossing segments under certain conditions using the above terminology, including a generalization of Lemma 1 from \cite{Linkage}.
\begin{figure}[h]
\begin{tikzpicture}[>=latex,scale=0.75,x=1in,y=1in] 
\node[anchor=south west] (label) at (0,0){
        \opt{color}{\includegraphics[width=4.5in]{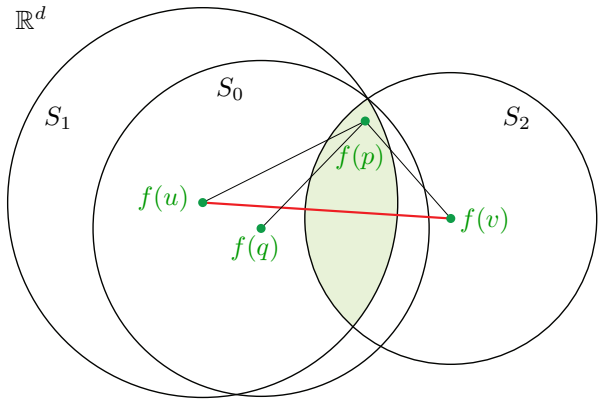}}
        \opt{gray}{\includegraphics[width=4.5in]{img/nonexpan_bw.eps}}
      };
\node[anchor=north west, color=black] at (1,3) {$\mathbb{R}^d$};
\node[anchor=north west, color=mygreen] at (1.85,1.75) {$f(u)$};
\node[anchor=north west, color=mygreen] at (2.5,1.4) {$f(q)$};
\node[anchor=north west, color=mygreen] at (4.1,1.6) {$f(v)$};
\node[anchor=north west, color=mygreen] at (3.23,2.04) {$f(p)$};
\node[anchor=north west, color=black] at (2.4,2.5) {$S_0$};
\node[anchor=north west, color=black] at (1.2,2.3) {$S_1$};
\node[anchor=north west, color=black] at (4.4,2.3) {$S_2$};
\end{tikzpicture}
\vspace{-0.5pc}
\caption{Points $f(u), f(v), f(q), f(p)$ with spheres $S_0, S_1, S_2$. The shaded area $S_1\cap S_2 \subset S_0$ is the region in which $f(p)$ may exist if $\{p,u,v\}$ is nonexpansive under $f$.}
\figlabel{fig:nonexpansive}
\end{figure}
\begin{lemma}
\lemlabel{lem:nonexpansive}
Consider distinct points $p,q,u,v\in\mathbb{R}^2$ with $p,u,v$ not collinear, line segment $(p,q)$ crossing line segment $(u,v)$, and a mapping $f : \{p,q,u,v\}\rightarrow\mathbb{R}^d$. (a) If $\{q,u,v\}$ is critical and $\{p,u,v\}$ is nonexpansive under $f$, then  $\{p,q\}$ is nonexpansive under $f$. (b) If $\{u,v\}$ is critical, and $\{p,u,v\},\{q,u,v\}$ are nonexpansive under $f$, then $\{p,q\}$ is nonexpansive under $f$; additionally if $\{p,q\}$ is critical under $f$, then $\{p,q,u,v\}$ is also. 
\end{lemma}

\begin{proof}
(a) Consider the following $d$-dimensional balls: $S_0$ centered at $f(q)$ with radius $\norm{p - q}$, $S_1$ centered at $f(u)$ with radius $\norm{p - u}$, and $S_2$ centered at $f(v)$ with radius $\norm{p - v}$ (see \figref{fig:nonexpansive}). $\{p,u,v\}$ nonexpansive under $f$ implies $f(p) \in S_1\cap S_2$. $\{q,u,v\}$ critical and $(p,q)$ crossing $(u,v)$ implies $S_1\cap S_2 \subset S_0$. Because $f(p)\in S_0$, $\{p,q\}$ is nonexpansive under $f$.

(b) Let $x = u+t(v-u)$ be the intersection of $(p,q)$ and $(u,v)$ and let $x_f = f(u) + t(f(v)-f(u))$.  Repeated application of \lemref{lem:nonexpansive}(a) yields $\norm{x-i} \geq \norm{x_f-f(i)}$ for $i\in\{p,q\}$. Combining with $\norm{x-p}+\norm{x-q}=\norm{p-q}$ and the triangle inequality, $\norm{x_f-f(p)}+\norm{x_f-f(q)}\geq \norm{f(p)-f(q)}$, yields $\{p,q\}$ nonexpansive under $f$. Further, if $\{p,q\}$ is critical under $f$, then so is $\{p,q,x_f\}$. Segments $(f(p),f(q))$ and $(f(u),f(v))$ are coplanar crossing at $x_f$ such that $\{u,p\}$ expansive implies $\{u,q\}$ contractive under $f$. Since $\{p,q,u,v\}$ is nonexpansive, $\{p,q,u,v\}$ must be critical under $f$. 
\end{proof}

We will consider a \emph{polygon} $P$ to be a bounded closed figure in $\mathbb{R}^2$ bounded by finitely many line segments connected in a simple cycle. This definition restricts polygons to topological disks, and allows adjacent edges to be collinear. Let $V(P)$ denote the vertices of $P$, $\partial P$ denote the boundary of $P$, with $V(P)\subset \partial P\subset P$. An edge of $P$ is a line segment in $\partial P$ with endpoints at adjacent vertices. We say that a point $p\in P$ is \emph{visible} from a vertex $v\in V(P)$ if the line segment from $p$ to $v$ is in $P$. With the terminology in place, we can now state the problem (see \figref{fig:problem}).
\begin{problem}
Given a polygon $P$ in the plane with a boundary mapping $f: \partial P\rightarrow \mathbb{R}^d$, find an isometric mapping $g: P\rightarrow \mathbb{R}^d$ such that $g(\partial P) = f(\partial P)$.
\end{problem}
\begin{figure}[h]
\begin{tikzpicture}[>=latex,scale=0.75,x=1in,y=1in] 
\node[anchor=south west] (label) at (0,0){
	\opt{color}{\includegraphics[width=4.5in]{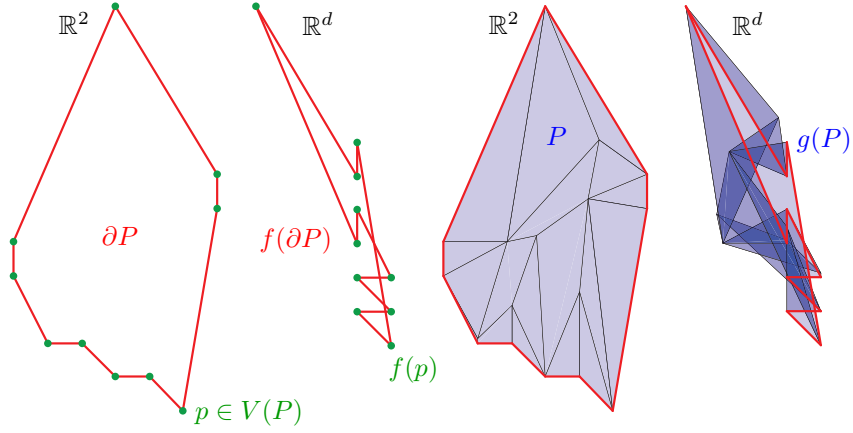}}
	\opt{gray}{\includegraphics[width=4.5in]{img/problem_bw.eps}}
      };
\node[anchor=north west, color=black] at (0.5,3) {$\mathbb{R}^2$};
\node[anchor=north west, color=black] at (2.2,3) {$\mathbb{R}^d$};
\node[anchor=north west, color=black] at (3.5,3) {$\mathbb{R}^2$};
\node[anchor=north west, color=black] at (5.2,3) {$\mathbb{R}^d$};
\node[anchor=north west, color=myred] at (0.8,1.5) {$\partial P$};
\node[anchor=north west, color=myred] at (1.9,1.5) {$f(\partial P)$};
\node[anchor=north west, color=myblue] at (3.9,2.2) {$P$};
\node[anchor=north west, color=myblue] at (5.65,2.2) {$g(P)$};
\node[anchor=north west, color=mygreen] at (1.45,0.3) {$p\in V(P)$};
\node[anchor=north west, color=mygreen] at (2.8,0.6) {$f(p)$};
\end{tikzpicture}
\vspace{-1pc}
\caption{Input and output to the hole problem showing notation.} 
\figlabel{fig:problem}
\end{figure}
If one exists, we call $g$ a \emph{solution} to the hole problem. Mapping $P$ into $\mathbb{R}$ requires infinitely many folds, so we restrict to $d\geq 2$ for the remainder.
\section{Necessary Condition}
\label{sec:necessary}

In this section, we define \emph{valid} boundary mappings and give a necessary condition for the hole problem under the weak assumption that the polygon boundary is folded at finitely many points.

\begin{definition}
\emph{(Valid Mapping)} Given polygon $P$ and boundary mapping $f:\partial P\rightarrow\mathbb{R}^d$, define $f$ to be \emph{valid} if $\partial P$ is nonexpansive under $f$ and adjacent vertices of $P$ are critical under $f$. 
\end{definition}

\begin{lemma}
\lemlabel{lem:straight}
Consider an instance of the hole problem with input polygon $P$ and boundary mapping $f:\partial P\rightarrow\mathbb{R}^d$ nonstraight at finitely many boundary points. If $f$ is not valid then the instance has no solution.
\end{lemma}
\begin{proof}
Modify $V(P)$ to include boundary points nonstraight under $f$ (vertices adjacent to collinear edges are allowed), so that $f$ is straight for $\partial P\setminus V(P)$. Assume a solution $g$ exists and $f$ is not valid. Then either two points $a,b\in\partial P$ are expansive under $f$, or two adjacent vertices $u,v\in V(P)$ are noncritical. If the former, then $\{a,b\}$ is also expansive under $g$, so $g$ cannot be isometric. If the latter, then $f(p)$ is nonstraight for some $p$ on the edge from $u$ to $v$, a contradiction.
\end{proof}

To determine the validity of $f$, checking expansiveness between all pairs of points in $\partial P$ is impractical. Instead it suffices to show that the set of vertices is nonexpansive under $f$, and edges of $P$ map to congruent line segments.

\begin{lemma}
\lemlabel{lem:valid}
Given polygon $P$ and boundary mapping $f:\partial P\rightarrow\mathbb{R}^d$, $f$ is valid if and only if $V(P)$ is nonexpansive and edges of $P$ map to congruent line segments under $f$.
\end{lemma}

\begin{proof}
If $f$ is valid, $ V(P)$ is nonexpansive under $f$ since $ V(P)\subset \partial P$, and edges map to congruent line segments because adjacent vertices are critical and points interior to edges are nonexpansive with endpoints. To prove the other direction, if edges of $P$ map to congruent line segments, adjacent vertices are critical and pairs of points on the same edge are nonexpansive (indeed critical) under $f$. To show that points from different edges are nonexpansive under $f$, consider vertex $p$ and point $q$ interior to the edge from vertex $u$ to $v$. By \lemref{lem:nonexpansive}(a), $\{q,p\}$ is nonexpansive under $f$ for any vertex $p$. Now consider point $q'\in \partial P$ not on the edge from $u$ to $v$. By the same argument as above, $\{q',u,v\}$ is nonexpansive under $f$, so by \lemref{lem:nonexpansive}(a), $\{q,q'\}$ is also nonexpansive. 
\end{proof}

\section{Bend Lines}
\label{sec:bend}

When the interior angle of the polygon boundary at a vertex decreases in magnitude under a valid boundary mapping, the local interior of the polygon will need to curve or bend to accommodate. For simplicity, we consider only single-fold solutions to satisfy such vertices, which will still be sufficient to construct a solution. We call these creases \emph{bend lines} made up of \emph{bend points}.

\begin{definition}
\emph{(Bend Points and Lines)} Given polygon $P$ with valid boundary mapping $f : \partial P\rightarrow \mathbb{R}^d$ and vertex $v\in V(P)$ adjacent to two vertices $\{u,w\}$ contractive under $f$, define $p\in P$ to be a \emph{bend point} of $(P,f,v)$ if there exists some $q\in \mathbb{R}^d$ (called a \emph{bend point image} of $p$) for which $\norm{p-i} = \norm{q-f(i)}$ for $i\in \{u,v,w\}$ and $p$ is visible from $v$. Further, define a \emph{bend line} of $(P,f,v)$ to be a maximal line segment of bend points of $(P,f,v)$, with one endpoint at $v$ and the other in $\partial P$; and let a \emph{bend line image} be a set of bend point images of the bend points in a bend line, congruent to the bend line. 
\end{definition}

A bend point corresponds to a point in the polygon such that triangles $\triangle pvu$ and $\triangle pvw$ isometrically map to triangles $\triangle qf(v)f(u)$ and $\triangle qf(v)f(w)$ respectively. Bend lines correspond to single folds of $P$ that locally satisfy isometry for the boundary from $u$ to $w$ through $v$. \lemref{lem:bend} represents bend points explicitly (see \figref{fig:bend}).
\begin{figure}[t]
\begin{tikzpicture}[>=latex,scale=0.75,x=1in,y=1in] 
\node[anchor=south west] (label) at (0,0){
	\opt{color}{\includegraphics[width=4.5in]{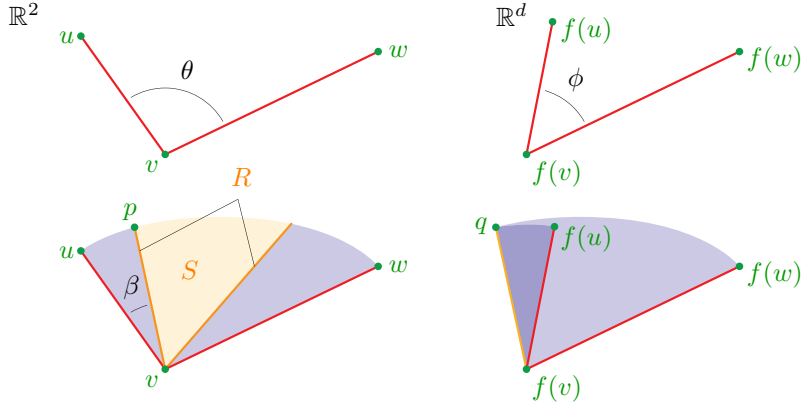}}
	\opt{gray}{\includegraphics[width=4.5in]{img/bend_bw.eps}}
      };
\node[anchor=north west, color=black] at (0.2,3) {$\mathbb{R}^2$};
\node[anchor=north west, color=black] at (3.6,3) {$\mathbb{R}^d$};
\node[anchor=north west, color=mygreen] at (0.55,2.8) {$u$};
\node[anchor=north west, color=mygreen] at (1.15,1.9) {$v$};
\node[anchor=north west, color=mygreen] at (2.85,2.7) {$w$};
\node[anchor=north west, color=mygreen] at (4.05,2.9) {$f(u)$};
\node[anchor=north west, color=mygreen] at (3.85,1.9) {$f(v)$};
\node[anchor=north west, color=mygreen] at (5.35,2.7) {$f(w)$};
\node[anchor=north west, color=black] at (1.4,2.6) {$\theta$};
\node[anchor=north west, color=black] at (4.1,2.55) {$\phi$};
\node[anchor=north west, color=black] at (1,1.1) {$\beta$};
\node[anchor=north west, color=myorange] at (1.4,1.2) {$S$};
\node[anchor=north west, color=myorange] at (1.75,1.85) {$R$};
\node[anchor=north west, color=mygreen] at (0.55,1.3) {$u$};
\node[anchor=north west, color=mygreen] at (1.15,0.4) {$v$};
\node[anchor=north west, color=mygreen] at (2.85,1.2) {$w$};
\node[anchor=north west, color=mygreen] at (1,1.6) {$p$};
\node[anchor=north west, color=mygreen] at (4.05,1.45) {$f(u)$};
\node[anchor=north west, color=mygreen] at (3.85,0.4) {$f(v)$};
\node[anchor=north west, color=mygreen] at (5.35,1.2) {$f(w)$};
\node[anchor=north west, color=mygreen] at (3.45,1.5) {$q$};
\end{tikzpicture}
\vspace{-0.5pc}
\caption{The bend points of $(P,f,v)$ showing relavent angles $\{\theta,\phi,\beta\}$, points $\{u,v,w,p,f(u),f(v),f(w),q\}$, and sets $\{R,S\}$. The upper figures show only the boundary mapping, while the lower images show filled, locally satisfying mappings of the interior.} 
\figlabel{fig:bend}
\end{figure}

\begin{lemma}
\lemlabel{lem:bend}
Consider polygon $P$ with valid boundary mapping $f : \partial P\rightarrow \mathbb{R}^d$ and vertex $v$ adjacent to two vertices $\{u,w\}$ contractive under $f$. Let  $\theta = \angle uvw$ be the internal angle of $P$ at $v$; let $\phi = \angle f(u)f(v)f(w)$; and let
\[R = \left\{p\in P \, \left| \, \begin{array}{l}\angle pvu \in\left\{\frac{\theta-\phi}{2},\frac{\theta+\phi}{2}\right\} \\ p\textnormal{ visible from }v\end{array} \right.\right\},\, S = \left\{p\in P \, \left| \, \begin{array}{l}\angle pvu \in\left[\frac{\theta-\phi}{2},\frac{\theta+\phi}{2}\right] \\ p\textnormal{ visible from }v\end{array} \right.\right\}.\] 
Then the set of bend points of $(P,f,v)$ is $R$ if $d=2$, and $S$ otherwise.
\end{lemma}

\begin{proof}
A point $p\in P$ visible from $v$ is a bend point of $(P,f,v)$ only if triangles $\triangle pvu, \triangle pvw$ are congruent to $\triangle qf(v)f(u), \triangle qf(v)f(w)$ respectively for some bend point image $q$ by definition. Let $\beta=\angle pvu$. If $d=2$, $\triangle pvu$ and $\triangle pvw$ must be coplanar. Then the internal angles of both triangles at $v$ must sum to $\theta$, and the magnitude of their difference $|(\theta-\beta) - \beta|$ must be  $\phi$. This condition is satisfied only when $\beta \in\left\{\frac{\theta-\phi}{2},\frac{\theta+\phi}{2}\right\}$. Thus for $d=2$, the set of bend points of $(P,f,v)$ is $R$.

For $d > 2$, triangles $\triangle qf(v)f(u), \triangle qf(v)f(w)$ need not be coplanar. Because $\{u,w\}$ is contractive under $f$, $\phi \geq |\theta-2\beta|$, so $\frac{\theta-\phi}{2} \leq \beta\leq \frac{\theta+\phi}{2}$, and points in $P\setminus S$ cannot be bend points. It remains to show that for each point $p\in S$ there exists a satisfying bend point image $q\in\mathbb{R}^d$. For a given $p$, $q$ must lie on two hyper-cones each with apex $v$, one symmetric about the segment from $f(v)$ to $f(u)$ with internal half angle $\beta$, and the other symmetric about the segment from $f(v)$ to $f(w)$ with internal half angle $\theta-\beta$. These hyper-cones have nonzero intersection $H$ because $(\theta-\beta)+\beta > \phi$ and $\phi \geq\max(\theta-\beta,\beta) -\min(\theta-\beta,\beta)$. The intersection of two hyper-cones with common apex $v$ is a set of rays emanating from $v$, so $H$ intersects the $(d-1)$-sphere centered at $f(v)$ with radius $\norm{p-v}$. Any point in this intersection satisfies all three constraints of a bend point image for any $p\in S$. 
 \end{proof} 

For every $d>2$, the set of bend points of $(P,f,v)$ is the same, but the set of bend point images increases with dimension. The set of bend point images is a ruled hypersurface of bend line images emanating from $f(v)$. In the case of $d=2$ above, hyper-cones are simply two rays, leading to disjoint line segments of bend points. For $d=3$ the set of bend points is a standard cone-like surface. Mapping generally to $\mathbb{R}^d$, the set is a ruled hypersurface of rays emanating from a point. 

\section{Split Points}
\label{sec:split}

Bend lines locally satisfy the boundary around a vertex with a single crease. We want to find the bend point on a bend line farthest from the vertex that remains nonexpansive with the rest of the boundary. We call such a point a \emph{split point}. 

\begin{definition}
\emph{(Split Points)} Given polygon $P$ with valid boundary mapping $f : \partial P\rightarrow \mathbb{R}^d$ and vertex $v$, contractive under $f$ with every visible nonadjacent vertex, adjacent to two vertices $\{u,w\}$ contractive under $f$, define $p$ to be a \emph{split point} of $(P,f,v)$, $q$ to be its \emph{split point image}, and $x$ to be its \emph{split end} if
\begin{enumerate}
\item $p$ is a bend point of $(P,f,v)$, with $q$ its bend point image;
\item $\norm{p-i} \geq \norm{q-f(i)}$ for $i\in V(P)$;
\item $\norm{p-x} = \norm{q-f(x)}$ for some $x\in V(P)\setminus\{u,v,w\}$; and
\item $p$ is visible from $x$.
\end{enumerate}
\end{definition}

\begin{lemma}
\lemlabel{lem:splitline}
Given polygon $P$ with valid boundary mapping $f : \partial P\rightarrow \mathbb{R}^d$ and vertex $v$ adjacent to two vertices $\{u,w\}$ contractive under $f$ with $v$ contractive under $f$ with any visible nonadjacent vertex, there exists a split point/image/end triple $(p,q,x)$ for every bend line/image pair $(L,L_f)$ of $(P,f,v)$ with $p\in L$ and $q\in L_f$. 
\end{lemma}

\begin{proof}
Given bend line/image pair $(L,L_f)$ we construct $(p,q,x)$. Parameterize $L$ so that $p(t)$ is the unique point in $L$ such that $\norm{p(t)-v} = t$ for $t\in[0,\ell]$ where $\ell$ is the length of $L$; and let $q(t)$ be the corresponding bend point image of $p(t)$ in $L_f$. For any $t\in[0,\ell]$ and vertex $x$, let $d(t,x) = \norm{p(t)-x}-\norm{q(t)-f(x)}$. Let $t^*$ be the maximum $t\in(0,\ell]$ for which $d(t,x) \geq 0$ for all $x\in V(P)$, and let $X$ be the set of such vertices $x\in V(P)\setminus\{u,v,w\}$ for which $d(t^*,x)=0$, and $d(t^*+\delta,i) < 0$ for all $\delta\in(0,\varepsilon]$ for some $\varepsilon>0$. If we can prove there exists some $x\in X$ from which $p(t^*)$ is visible, then $p = p(t^*)$ is a split point with $q = q(t^*)$ its split point image, satisfying the split point conditions by construction. 

Suppose for contradiction that $t^*$ does not exist so that for all $t\in(0,\ell]$, $d(t,x) < 0$ for some $x\in V(P)$. Because $d$ is continuous and $d(0,x) \geq 0$ for all $x\in V(P)$, there exists a vertex $x'\in V(P)\setminus\{u,v,w\}$ not visible from and critical with $v$ under $f$ such that $d(\delta,x') < 0$ for all $\delta \in (0,\varepsilon]$ for some $\varepsilon>0$. Either $x'$ is in the infinite sector $C$ induced by $\angle uvw$ or not. If the former, the line segment from $v$ to $x'$ must cross some edge $(a,b)$ of $P$ and $\{a,b,x',v\}$ is critical under $f$ by \lemref{lem:nonexpansive}(b). Since neither $a$ nor $b$ can be visible from $v$, then $u$ and $w$ must be in $\triangle abv$, and $\{u,v,w\}$ must be critical, contradicting $\{u,w\}$ contractive under $f$. Alternatively $x'$ is not in $C$, and for every $\delta\in(0,\varepsilon]$ for some $\varepsilon>0$, the line segment from $p(\delta)$ to $x'$ crosses either $(v,u)$ or $(v,w)$. By \lemref{lem:nonexpansive}(b), $d(\delta,x') \geq 0$, a contridiction, so $t^*$ exists.
\begin{figure}[t]
\begin{tikzpicture}[>=latex,scale=0.75,x=1in,y=1in] 
\node[anchor=south west] (label) at (0,0){
	\opt{color}{\includegraphics[width=4.5in]{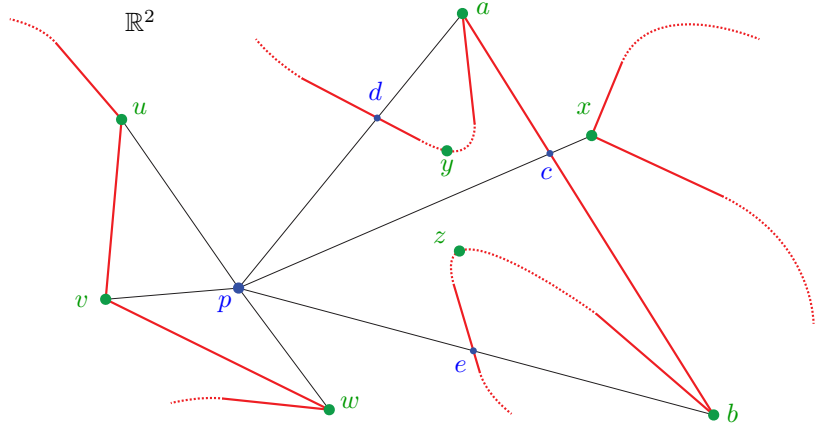}}
	\opt{gray}{\includegraphics[width=4.5in]{img/visible_bw.eps}}
      };
\node[anchor=north west, color=black] at (1,3) {$\mathbb{R}^2$};
\node[anchor=north west, color=mygreen] at (0.65,1) {$v$};
\node[anchor=north west, color=mygreen] at (1.05,2.35) {$u$};
\node[anchor=north west, color=mygreen] at (2.5,0.3) {$w$};
\node[anchor=north west, color=myblue] at (1.65,1.00) {$p$};
\node[anchor=north west, color=mygreen] at (4.15,2.35) {$x$};
\node[anchor=north west, color=mygreen] at (3.45,3.05) {$a$};
\node[anchor=north west, color=myblue] at (3.9,1.9) {$c$};
\node[anchor=north west, color=mygreen] at (5.2,0.25) {$b$};
\node[anchor=north west, color=myblue] at (2.7,2.5) {$d$};
\node[anchor=north west, color=mygreen] at (3.2,1.95) {$y$};
\node[anchor=north west, color=myblue] at (3.3,0.55) {$e$};
\node[anchor=north west, color=mygreen] at (3.15,1.45) {$z$};
\end{tikzpicture}
\caption{Visibility of $p$. If $x\in X$ is not visible from $v$, one of $\{a,b,y,z\}\in X$ will be.}
\figlabel{fig:visible}
\vspace{-0.7pc}
\end{figure}

We now prove that $p$ is visible from some $x\in X$. Suppose for contradiction that $p$ is not visible from any $x\in X$ so that for each $x$ there exists point $c\in\partial P$, the boundary crossing closest to $p$ on the segment from $p$ to $x$. $c$ cannot be strictly interior to edge $(v,u)$ or $(v,w)$ because \lemref{lem:nonexpansive}(b) implies $\norm{p(t^*+\delta)-x} = \norm{q(t^*+\delta)-f(x)}$ for all $d\in(0,\varepsilon]$ for some $\varepsilon$, a contridiction. And $c$ cannot be $v$ or else $\norm{p(t)-x} = \norm{q(t)-f(x)}$ for all $t\in[0,\ell]$. So $c$ crosses some other edge $(a,b)$ (see \figref{fig:visible}). Then \lemref{lem:nonexpansive}(b) implies $\norm{p-i} = \norm{q-f(i)}$ for $i\in\{a,b\}$, and the contrapositive of \lemref{lem:nonexpansive}(a) implies for at least one vertex $i\in\{a,b\}$, $\norm{p(t^*+\delta)-i} < \norm{q(t^*+\delta)-f(i)}$ for all $\delta\in(0,\varepsilon]$ for some $\varepsilon>0$. Without loss of generality assume $i=a$. Because $a\in X$, $p$ cannot be visible from $a$. Let $d\in\partial P$ be the boundary crossing closest to $p$ on the segment from $p$ to $a$.  There must exist some vertex $y$ in triangle $\triangle acp$ from which $p$ is visible because the boundary of the polygon entering the triangle at $d$ must return to $a$ without crossing edge $(c,p)$. By the same argument, at least one of $\{y,b\}$ is in $X$, and since $p$ is visible from $y$, $b\in X$. Replacing $(b,e,z)$ for $(a,d,y)$ in the argument above, one of $\{y,z\}$ is in $X$. But $p$ is visible from both, a contradiction.
\end{proof}
 
 \begin{lemma}
 \lemlabel{lem:splitpolygon}
Given polygon $P$ with valid boundary mapping $f : \partial P\rightarrow \mathbb{R}^d$ and vertex $v$, contractive under $f$ with every visible nonadjacent vertex, adjacent to two vertices $\{u,w\}$ contractive under $f$, a split point/image/end triple of $(P,f,v)$ exists and can be identified in $O(d|V(P)|)$ time.
 \end{lemma}
 \begin{proof}
This result follows directly by choosing any bend line/image pair of $(P,f,v)$ according to \lemref{lem:bend}, then constructing the split point/image/end triple specified by \lemref{lem:splitline}. Choosing a bend line/image pair can be done in $O(d)$ time. Constructing the split point/image/end triple requires a $d$-dimensional comparison at each vertex yielding total construction time $O(d|V(P)|)$.
 \end{proof}

\section{Partitions}
\label{sec:partition}

To find an overall solution to the hole problem, we will repeatedly split a polygon in half, solve each piece recursively, and then join the pieces back together. Specifically, we want to find a \emph{partition} consisting of two \emph{partition polygons} together with respective boundary mappings such that: the partition polygons exactly cover the original polygon; the partition polygons intersect, and only on their boundaries; each partition function maps the partition polygon boundaries into the same dimensional space as the original function;  the original boundary mapping of the polygon boundary is preserved by the partition functions; the intersection of the partition polygons map to the same place under both partition functions; and the partition functions are valid.

\begin{definition}
\emph{(Valid Partition)} Given polygon $P$ and valid boundary mapping $f : \partial P\rightarrow \mathbb{R}^d$, define $(P_1,P_2,f_1,f_2)$ to be a \emph{valid partition} of $(P,f)$ if the following properties hold:
\[\begin{array}{ll}
(1)\quad P_1,P_2\textnormal{ polygons with }P = P_1 \cup P_2; \quad & (4)\quad P_1\cap P_2 = \partial P_1\cap\partial P_2= L\neq \emptyset; \\
(2)\quad f_1:\partial P_1\rightarrow\mathbb{R}^d, f_2:\partial P_2\rightarrow\mathbb{R}^d; & (5) \quad f(p) = \cases{f_1(p) &  p\in \partial P\cap \partial P_1, \\ f_2(p) & \textnormal{otherwise;}} \\
(3)\quad f_1(p) = f_2(p) \textnormal{ for } p\in L; & (6)\quad f_1,f_2 \textnormal{ valid.}
\end{array}\]
\end{definition}

\section{Algorithm}
\label{sec:theorem}

\begin{theorem}
\thmlabel{main}
Given polygon $P$ and boundary mapping $f : \partial P\rightarrow \mathbb{R}^d, d \geq 2$ nonstraight at finitely many boundary points, an isometric mapping $g:P\rightarrow\mathbb{R}^d$ with $g(\partial P) = f(\partial P)$ exists if and only if $f$ is valid. A solution can be computed in polynomial time.
\end{theorem}

The theorem implies that the necessary condition in \lemref{lem:straight} is also sufficient. Our approach is to iteratively divide $P$ into valid partitions and combine them back together. We partition non-triangular polygons into smaller ones differently depending on which of two properties $(P,f)$ satisfies. First we show that $(P,f)$ satisfies at least one of these properties.

\begin{lemma}
\lemlabel{lem:existcondition}
For every polygon $P$ with $| V(P)| > 3$ and valid boundary mapping $f : \partial P\rightarrow \mathbb{R}^d$, either (a) there exist two nonadjacent vertices $\{u,v\}$ critical under $f$ and visible from each other, or (b) there exists a vertex $v\in V(P)$ adjacent to two vertices $\{u,w\}$ contractive under $f$, or (c) both exist.
\end{lemma}

\begin{proof}
Suppose for contradiction that there exists some $(P,f)$ such that no two nonadjacent vertices critical under $f$ are visible from each other and no vertex is adjacent to two vertices contractive under $f$. Consider any vertex $v$ which, by the contrapositive of the latter condition, will be adjacent to two vertices $\{u,w\}$ critical under $f$. Since  $|V(P)| > 3$, $u$ and $v$ are nonadjacent and cannot be visible from each other, so there must be at least one other distinct vertex $x$ interior to $\triangle uvw$ visible from vertex $v$. But since $x$ is nonexpansive with $\{u,v,w\}$ under $f$, $\{x,u,v,w\}$ must be critical under $f$, a contradiction.
\end{proof}

\begin{lemma}
\lemlabel{lem:critical}
Consider polygon $P$ with valid boundary mapping $f : \partial P\rightarrow \mathbb{R}^d$ containing nonadjacent vertices $\{u,v\}$ critical under $f$ with $u$ visible from $v$. Construct polygon $P_1$ from the vertices of $P$ from $u$ to $v$, and $P_2$ from the vertices of $P$ from $v$ to $u$. Construct boundary mapping functions $f_1:\partial P_1\rightarrow\mathbb{R}^d$, $f_2:\partial P_2\rightarrow\mathbb{R}^d$ so that $f_1(x) = f(x)$ for $x\in V(P_1)$, $f_2(x) = f(x)$ for $x\in V(P_2)$, with $f_1,f_2$ mapping edges of $P_1,P_2$ to congruent line segments. Then $(P_1,P_2,f_1,f_2)$ is a valid partition.
\end{lemma}

\begin{proof}
Because $P_1$ and $P_2$ are constructed by splitting $P$ along line segment $L\subset P$ from $u$ to $v$, $P = P_1\cup P_2$ and $L = P_1\cap P_2 = \partial P_1\cap\partial P_2$, satisfying properties (1) and (2) of a valid partition. Property (3) is satisfied by definition. Property (4) holds because $f$ is valid, $\{u,v\}$ is critical, and points in $L$ are nonexpansive with points in $\partial P_1$ and $\partial P_2$ by \lemref{lem:nonexpansive}(a). Property (5) holds by construction. Lastly, Property (6) holds because $f_1,f_2$ satisfy the conditions in \lemref{lem:valid} by construction.
\end{proof}

\begin{lemma}
\lemlabel{lem:noncritical}
Consider polygon $P$ with valid boundary mapping $f : \partial P\rightarrow \mathbb{R}^d$ and vertex $v\in V(P)$, contractive under $f$ with every visible nonadjacent vertex, adjacent to two vertices $\{u,w\}$ contractive under $f$. Let $(p,q,x)$ be a split point/image/end triple of $(P,f,v)$. Construct polygon $P_1$ from $p$ and the vertices of $P$ from $v$ to $x$, and $P_2$ from $p$ and the vertices of $P$ from $x$ to $v$. Construct boundary mapping functions $f_1:\partial P_1\rightarrow\mathbb{R}^d$, $f_2:\partial P_2\rightarrow\mathbb{R}^d$ so that $f_1(x) = f(x)$ for $x\in V(P_1)\setminus p$, $f_2(x) = f(x)$ for $x\in V(P_2)\setminus p$, $f_1(p) = f_2(p) = q$, with $f_1,f_2$ mapping edges of $P_1,P_2$ to congruent line segments. Then $(P_1,P_2,f_1,f_2)$ is a valid partition. 
\end{lemma}

\begin{proof}
Because $P_1$ and $P_2$ are constructed by splitting $P$ along two line segments fully contained in $P$, $P = P_1\cup P_2$ and $P_1\cap P_2 = \partial P_1\cap\partial P_2$, satisfying properties (1) and (2) of a valid partition. Property (3) is satisfied by definition. Property (4) holds because $(P,f)$ is valid, $V(P_1)$ and $V(P_2)$ are nonexpansive, with adjacent vertices critical under $f$ by definition of a split point/image, and points in the new line segments are nonexpansive with points in $\partial P_1$ and $\partial P_2$ by \lemref{lem:nonexpansive}(a). Property (5) holds by construction. Lastly, Property (6) holds because $f_1,f_2$ satisfy the conditions in \lemref{lem:valid} by construction.
\end{proof}

Next, we establish the base case for our induction. Specifically a triangle with a valid boundary mapping of its boundary has a unique isometric mapping of its interior consistent with the provided boundary condition.

\begin{lemma}
\lemlabel{lem:triangle}
Given polygon $P$ with $| V(P)| = 3$ and valid boundary mapping $f : \partial P\rightarrow \mathbb{R}^d$, there exists a unique isometric mapping $g: P\rightarrow\mathbb{R}^d$ such that $g(B) = f(B)$. 
\end{lemma}

\begin{proof}
Because $f$ is valid, the vertices of $P$ are critical under $f$. $\partial P$ and $f(\partial P)$ are congruent triangles, so their convex hulls are isometric. Specifically, if $P$ with vertices $\{u,v,w\}$ is parameterized by $P = \{p(a,b) = a(v-u)+b(w-u)+u \,|\, a,b\in[0,1], a+b\leq1\}$, then the affine map $g: P\rightarrow \mathbb{R}^d$ defined by 
\[g(p(a,b)\in P) = a[f(v)-f(u)]+b[f(w)-f(u)]+f(u)\]
is a unique isometry for $g(B) = f(B)$.
\end{proof}

Lastly we show that we can combine isometric mappings of valid partitions into larger isometric mappings.

\begin{lemma}
\lemlabel{lem:combine}
Consider polygon $P$ with valid boundary mapping $f : \partial P\rightarrow \mathbb{R}^d$, with valid partition $(P_1,P_2,f_1,f_2)$. Given isometric mappings $g_1: P_1\rightarrow \mathbb{R}^d$, $g_2: P_2\rightarrow \mathbb{R}^d$ with $g_1(\partial P_1) = f_1(\partial P_1)$, $g_2(\partial P_2) = f_2(\partial P_2)$, the mapping $g : P \rightarrow\mathbb{R}^d$ defined below is also isometric, with $g(\partial P) = f(\partial P)$: 
\[g(p\in P) = \cases{g_1(p) &p\in P_1, \\ g_2(p) &\textnormal{otherwise.} }\]
\end{lemma}

\begin{proof}
First, $g(\partial P) = f(\partial P)$  because the partition is valid. Consider the shortest path $K$ between points $p,q\in P$ composed from a finite set of line segments. Suppose for contradiction that $g(K)$ is not the same length as $K$. Every point in $K$ either lies in $P_1$, $P_2$, or both by property (1) of a valid partition. Split $K$ into a connected set of line segments, each segment fully contained in either $P_1$ or $P_2$ with endpoints in $P_1\cap P_2$. Because $g_1$ and $g_2$ are isometric, these line segments remain the same length under $g$. Further, the endpoints of adjacent segments map to the same place under $g_1$ and $g_2$ by definition of a valid parition. The total length of $g(K)$ is the sum of the lengths of the intervals, the same length as $K$, a contradiction. 
\end{proof}

Now we are ready to prove the theorem.

\begin{proof}
\lemref{lem:straight} implies that $f$ is valid if $g$ exists. We show $g$ exists for valid $f$ by construction. Partition $(P,f)$ with $|V(P)| > 3$ as follows. If $(P,f)$ contains two nonadjacent vertices $\{u,v\}$ critical under $f$ and visible from each other, divide using Routine 1: partition using the construction in \lemref{lem:critical}. Otherwise divide using Routine 2: partition using the construction in \lemref{lem:noncritical}, applying Routine 1 to each partitioned polygon immediately after. Note that both polygons generated by the construction from \lemref{lem:noncritical} are guaranteed to contain two nonadjacent vertices critical under $f$ and visible from each other, namely $\{u,p\}$ and $\{w,p\}$, so each can be divided using Routine 1. Recursively fill each partitioned polygon with an isometric mapping of their interior and combine them into a mapping $g : P\rightarrow \mathbb{R}^d$ using the construction in \lemref{lem:combine}. Since the partitions are valid, $g$ is isometric with $g(\partial P) = f(\partial P)$. Construct isometries for triangular polygons, the base case of the recursion, according to \lemref{lem:triangle}. 

To show the recursion terminates, consider state $i$ where $P$ is partitioned into a set of $n_i$ polygons $\mathcal{P}_i = \{P_1,\ldots,P_{n_i}\}$. Define potential $\Phi_i = \sum_{P_j\in\mathcal{P}_i}(|V(P_j)| - 3)$ with $\Phi_0 = |V(P)| - 3$. Partitioning a polygon using Routine 1 yields state $i+1$ with $\Phi_{i+1} = \Phi_i - 1$: \lemref{lem:critical} adds two vertices, the number of polygons increases by one, and $2-3 = -1$. Partitioning a polygon using Routine 2 also yields $\Phi_{i+1} =  \Phi_i - 1$: \lemref{lem:noncritical} adds four vertices, \lemref{lem:critical} adds two vertices with each application, the number of polygons increases by three, and $4+2\times 2-3\times 3 = -1$. \lemref{lem:existcondition} ensures that one of the routines can always be applied to any non-triangular polygon. When $\Phi_i = 0$, all partitioned polygons are triangles and no polygon can be partitioned further. The iteration terminates after $\Phi_0$ calls to either routine. 

Let $n$ be the number of vertices $|V(P)|$ in the input polygon. At the start of the algorithm, all critical vertex pairs can be identified naively in $O(dn^2)$ time. Application of either routine requires at most $O(dn)$ time, and both routines can update and maintain new critical vertex pairs in partition polygons at no additional cost. Each routine is called no more than $O(n)$ times. Only a linear number of triangles are produced and the construction of each $g_i$ takes constant time. The running time of the entire construction is thus $O(dn^2)$, which is polynomial. 
\end{proof}

\begin{figure}[h]
\includegraphics[width=4.5in]{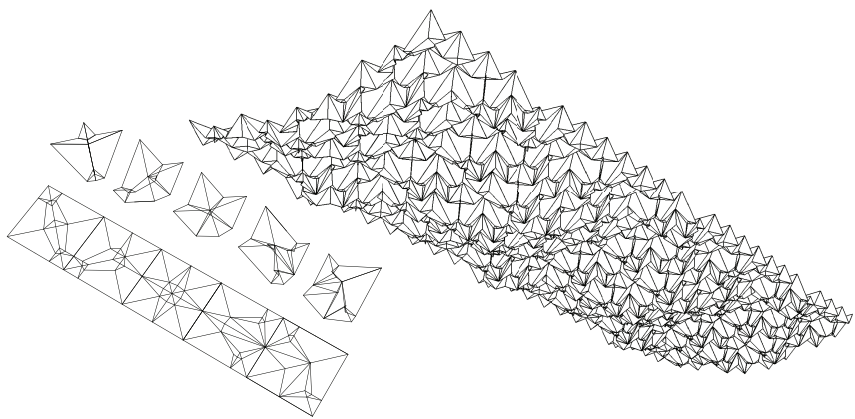}
\caption{A three-dimensional abstract tessellation formed by tiling five different square units, each corner in either a binary low or high state. Units were designed using this algorithm having common boundaries, connected to form single sheet tessellations.}
\figlabel{fig:surface}
\end{figure}
\begin{figure}[h]
\opt{color}{\includegraphics[width=4.5in]{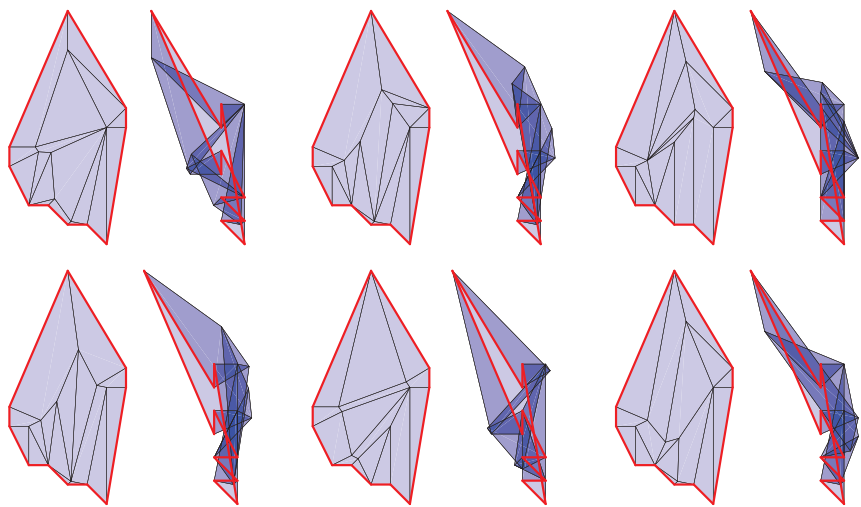}}
\opt{gray}{\includegraphics[width=4.5in]{img/matlab_bw.eps}}
\caption{Various solutions for the same input polygon and boundary mapping found by our MATLAB implementation for $d=2$.}
\figlabel{fig:matlab}
\end{figure}

\section{Applications}
\label{sec:application}

Much of the intuition for this algorithm was developed while working on the design of various three-dimensional tessellations, specifically Maze Folding \cite{MazeFolding_Origami5} and a private commission designing an origami chandelier for Moksa, a restaurant in Cambridge, MA (see \figref{fig:surface}). A version of this algorithm was implemented for flat-folds $(d=2)$ in 2010 using MATLAB (see \figref{fig:matlab}). We leave an implementation of this algorithm in 3D for future work. 

\section{Conclusion}
\label{sec:conclusion}

We have proposed an algorithm for finding isometric mappings consistent with prescribed boundary mappings that runs in polynomial time. This algorithm was inspired by the universal molecule construction; instead of insetting an input polygon perimeter at a constant rate from all edges at once, our algorithm insets each vertex serially as far as possible. Our construction cannot find all possible isometric solutions, though the algorithm provides a rich family of solutions given choice of bend line and image with each application of Routine 2: two choices when $d=2$ and an infinite set of choices for $d>2$. This algorithm can be generalized by not insetting vertices all the way to split points, and by solving vertices locally with more than one crease at a vertex. We conjecture that adding such flexibility would allow construction of the entire space of isometric solutions following a similar procedure to our construction. 

Recall that the proposed algorithm does not address self intersection and cannot guarantee the existence of a valid layer ordering for the isometries found; however, because the space of solutions is large for a generic input, one might be able to construct non-self-intersecting solutions by directing the algorithm's decisions appropriately through the solution space. Additionally, the proposed algorithm only addresses instances for $f$ folded at finitely many points. It is conceivable that a similar algorithm could be used to design curved foldings. We leave these as open problems.

\section*{Acknowledgments}

The authors would like to thank Barry Hayes for introducing us to this problem, and to Robert Lang and Tomohiro Tachi for helpful discussions. E. Demaine supported in part by NSF ODISSEI grant EFRI-1240383 and NSF Expedition grant CCF-1138967.

\bibliographystyle{alpha}
\bibliography{hole_ku}

\end{document}